\documentclass[graybox]{svmult}

\usepackage{mathptmx}       
\usepackage{helvet}         
\usepackage{courier}        

\usepackage{type1cm}        
\usepackage{makeidx}         
\usepackage{graphicx}        
\usepackage{multicol}        
\usepackage[bottom]{footmisc}
\usepackage{float}

  \usepackage[numbers,square]{natbib}
  \usepackage{mathrsfs}
  \usepackage{amsmath}
  \usepackage{amssymb}
  \usepackage[utf8]{inputenc}    
\usepackage{hyperref}
  
  \usepackage{tikz}
 \usetikzlibrary{calc}
\usetikzlibrary{matrix}
\usepackage{pgfplots}
  \pgfplotsset{compat=1.13}

\spnewtheorem{thm}{Theorem}[section]{\bf}{\it}
\spnewtheorem{prop}[thm]{Proposition}{\bf}{\it}
\spnewtheorem{lem}[thm]{Lemma}{\bf}{\it}
\spnewtheorem{cor}[thm]{Corollary}{\bf}{\it}

\spnewtheorem{cond}[thm]{Condition}{\bf}{}
\spnewtheorem{exm}[thm]{Example}{\bf}{}
\spnewtheorem{rem}[thm]{Remark}{\bf}{}

\renewcommand{\phi}{\varphi}
\newcommand{\eps}{\varepsilon}

\newcommand{\ud}{\mathrm{d}}
\newcommand{\uod}{{\mathrm{od}}}

\newcommand{\R}{\mathbb{R}}

\newcommand{\N}{\mathbb{N}}

\newcommand{\Lz}{L^2}
\newcommand{\dom}{\mathfrak{D}}

\newcommand{\norm}[1]{\ensuremath{\left\Vert #1 \right\Vert}}
\newcommand{\abs}[1]{\ensuremath{\left\vert #1 \right\vert}}

\newcommand{\hilb}{\mathcal{H}}
\newcommand{\uppar}[1]{\ensuremath{^{(#1)}}}
\smartqed
  
\makeindex             

\begin{document}

\title*{The Massless Nelson Hamiltonian and its Domain}
\author{Julian Schmidt}
\institute{Julian Schmidt \at Fachbereich Mathematik, Eberhard Karls Universit\"at T\"ubingen, Auf der Morgenstelle 10, 72076 T\"ubingen, Germany. \email{juls@maphy.uni-tuebingen.de}}

%
%
\maketitle

\abstract*{In the theory of point interactions, one is given a formal expression for a quantum mechanical Hamiltonian. The interaction terms of the Hamiltonian are singular: they can not be rigorously defined as a perturbation (in the operator or form sense) of an unperturbed free operator. A similar situation occurs in Quantum Field Theory, where it is known as the ultraviolet problem. Recently, it was shown that some of the tools used in the context of point interactions can be adapted to solve the problem of directly defining a Hamiltonian for the Nelson model. This model provides a well studied example of a bosonic quantum field that is linearly coupled to nonrelativistic particles. The novel method employs so called abstract interior-boundary conditions to explicitly characterise the action and the domain of the Hamiltonian without the need for a renormalisation procedure. Here, for the first time, the method of interior-boundary conditions is applied to the massless Nelson model. Neither ultraviolet nor infrared cutoffs are needed.}

\abstract{In the theory of point interactions, one is given a formal expression for a quantum mechanical Hamiltonian. The interaction terms of the Hamiltonian are singular: they can not be rigorously defined as a perturbation (in the operator or form sense) of an unperturbed free operator. A similar situation occurs in Quantum Field Theory, where it is known as the ultraviolet problem. Recently, it was shown that some of the tools used in the context of point interactions can be adapted to solve the problem of directly defining a Hamiltonian for the Nelson model. This model provides a well studied example of a bosonic quantum field that is linearly coupled to nonrelativistic particles. The novel method employs so called abstract interior-boundary conditions to explicitly characterise the action and the domain of the Hamiltonian without the need for a renormalisation procedure. Here, for the first time, the method of interior-boundary conditions is applied to the massless Nelson model. Neither ultraviolet nor infrared cutoffs are needed.}

\section{Introduction}
\label{sec:1}
In this contribution we will discuss how some of the tools that have been developed in the theory of (many body-)point interactions can be adapted to define Hamiltonians for certain models of Quantum Field Theory. In these models, a nonrelativistic particle interacts linearly with a bosonic quantum field, which means that the interaction term in a formal Hamiltonian is linear in creation and annihilation operators. If one wants to set up a self-adjoint Hamiltonian for such a model, the main obstacle is the fact that this interaction term is in general not small -- in the operator or form sense -- relative to the free operator $L$, i.e. the Hamiltonian for the non-interacting system of particles and field. Because the relative bound is given by an integral in Fourier space, which does or does not converge for large momenta, this is also called the \textit{ultraviolet problem}. Well studied examples with linear coupling are the so called massive and massless Nelson models. Until recently, the standard approach to overcome the ultraviolet problem was a renormalisation procedure, where the interaction is restricted by hand to momenta $\abs k \leq \Lambda$ for some positive $\Lambda$ in order to render the bound finite. This \textit{UV-cutoff} results in a self-adjoint cutoff Hamiltonian $H_\Lambda$. In some models, including the massive and the massless Nelson model, there exists a diverging sequence of so called renormalisation constants $E_\Lambda$ such that $H_\Lambda+E_\Lambda$ converges for $\Lambda \rightarrow \infty$ in norm resolvent sense to a self-adjoint operator $H_\infty$. This is called \textit{removing the UV-cutoff} and the operator $H_\infty$ is called the renormalised Hamiltonian. While the renormalisation method yields that the so obtained operator is bounded from below, neither the action of $H_\infty$ nor its domain $D(H_\infty)$ are obtained in this way. That is why, at the end of his seminal article of 1964, after carrying out the renormalisation procedure sketched above, Edward Nelson posed the following questions:
\begin{quotation}
It would be interesting to have a direct description
of the operator $H_\infty$. \\ Is $D(H_\infty) \cap D(L^{1/2}) = \{0 \}$? \ (\cite{nelson1964})
\end{quotation}
In the article \cite{GrWu17}, Griesemer and Wünsch finally gave the answer to the second question: Yes, in fact it even holds form the form domain that $D(\abs{H_\infty}^{1/2}) \cap D(L^{1/2}) = \{0 \}$. This was proved with the help of the renormalisation technique. While their result solved the second part of the problem posed by Nelson, it also showed the limitations of this method, for it required considerable technical effort to extract this information.

In the recent article \cite{nelsontype}, Jonas Lampart together with the author gave a complete answer to Nelson's question in the above quote. That is, to provide a direct desription of the operator $H_\infty$ and its domain, from which the answer to the second question can easily be read off. More concretely, a dense domain $D(H)$ on Fock space is constructed, whose elements are the sum of a regular part, which is an element of $D(L)$, and a singular part. Then the action of $L$ is extended to this domain in such a way that it encodes the action of the creation operator. In addition, also the action of the annihilation operator is extended to the domain $D(H)$ and it is shown that their sum defines a self-adjoint operator $H$, bounded from below. Afterwards it turns out, that this operator is in fact the limit of the sequence of cutoff operators $H_\Lambda$, so it becomes clear that $H$ is equal to the renormalised Hamiltonian $H_\infty$. 

Characterising elements of $D(H)$ in this way can be viewed as imposing abstract boundary conditions on them. These boundary conditions, which are called \textit{interior-boundary conditions}, are formulated in strong analogy with the theory of point interactions. The main difference being the fact that the boundary space or space of charges of the theory of point interactions is on each sector of Fock space identified with the sector with one boson less. In this way the boundary space can be identified with the Hilbert space $\hilb$ itself. The singular behaviour of the wave function on one sector is determined by the wave function one sector below. The Skornyakov--Ter-Martyrosyan (STM) operator appears in this construction not as part of a boundary condition and it is therefore not used to label self-adjoint realisations, for the latter alternative see, e.g. \cite{MiOt17}. Instead, the STM operator $T$ is identified as the correct extension of the annihilation operator to the singular functions and is therefore part of the action of the Hamiltonian. Thus it is not necessary to study $T$ as an operator on the space of charges, but as an operator on $\hilb$.

In Nelsons original work~\cite{nelson1964}, the so called massive case was treated, where the dispersion relation of the bosonic field is given by $(\abs{k}^2+m^2)^{1/2}$ for some $m>0$. Later, the renormalisation procedure was applied also to the massless case $m=0$ and the properties of the Hamiltonians with and without cutoff were investigated, see e.g. \cite{froehlich1974,{pizzo03},BDP12,MaMo17}. The result of Griesemer and Wünsch equally holds for the massless case.

In \cite{nelsontype}, the case of nonrelativistic particles was considered. In \cite{pseudorel18}, the construction was extended to treat also pseudorelativistic models with dispersion relations $\Theta(p) = \sqrt{p^2+\mu^2}$. If the renormalisation constant $E_\Lambda$ diverges too fast, the method of \cite{nelsontype} has to be suitably modified. This was done for the first time in \cite{La18}. In~\cite{DeltaPolaron}, the enhanced method of the former article is applied to a Polaron-type model. 

So far however, these results on interior-boundary conditions were concerned with the massive case: it was always assumed that the dispersion of the bosons is bounded from below by a positive constant. As a consequence, the free operator is bounded from below by the number (of bosons) operator, i.e. $N \leq L$. Now naturally the question arises whether the construction using abstract interior boundary conditions can be extended also to the massless case. After all, within renormalisation schemes, there is no difficulty in treating these cases as well. 

In the present note, we will give a more detailed description of the domain $D(H)$ with or without mass. Roughly speaking, we will differentiate Nelson's second question between the full free operator $L$ and the part of it that only acts on the field degrees of freedom, $\ud \Gamma(\omega)$. In this way, we will prove self-adjointness of the Hamiltonian $H$ with or without mass. Neither an ultraviolet nor an infrared cutoff is used in the construction, not even in an intermediate step. We will focus on a class of models in three space dimensions where one nonrelativistic particle interacts with the bosonic field.

In \cite{LiNi18}, interior boundary conditions were used in a multi-time formulation for massless Dirac particles in one space dimension. There the number of particles is bounded. As we will explain in more detail below, the main problems with massless fields occur only if an arbitrary number of quanta is allowed.

For physical aspects and more general discussions of the IBC approach, we refer the reader to \cite{{KeSi16},{traj18},{timeasymmetry}} and \cite{{TeTu16}}.

\section{The Model}
\label{sec:model}
In this section we will define the basic objects of our model. Then we will introduce a spectral parameter and justify its use by demonstrating that the domain and the extended annihilation operator are actually parameter-independent. 

Our model will be defined on the Hilbert space
\begin{align*}
\hilb: = \bigoplus_{n=0}^\infty \Lz(\R^3) \otimes \Lz_{\mathrm{sym}}(\R^{3n})
\end{align*}
of the composite system of the particle and the field. We will formulate the model in Fourier representation where elements of the sectors of this Hilbert space are wavefunctions 
\begin{align*}
\psi \uppar n(p,k_1, \dots, k_n) \, ,
\end{align*}
which are symmetric under exchange of either two of the $k$-variables. The operator that governs the dynamics of the nonrelativistic particle is given by the multiplication operator $p^2$. The dispersion relation of the field is given by a non-negative function $\omega \in L^\infty_{\mathrm{loc}}(\R^3)$. Its second quantisation will be denoted by $\Omega:=\ud \Gamma(\omega)$. We can now define the free operator $L=p^2+\Omega$, which is self-adjoint and non-negative with domain $D(L) \subset \hilb$. Since $\Omega \geq 0$, the operator $\Omega_\mu := \Omega+\mu$ is invertible for any $\mu >0$ and so is $L_\mu := p^2+ \Omega_\mu$. 

The interaction between the field and the particle is characterised by a coupling function $v \in \Lz_{\mathrm{loc}}(\R^3)$, which is called the \textit{form factor}. The formal expression for a Hamiltonian of the model is
\begin{align*}
L + a(V) + a^*(V) \, ,
\end{align*}
where the annihilation operator $a(V)$ acts sector-wise as
\begin{align*}
(a(V) \psi )\uppar n(p, k_1, \dots, k_n) := \sqrt{n+1} \int_{\R^3} \overline{v(k)} \psi \uppar {n+1}(p-k,k_1, \dots, k_n , k) \, \ud k \, .
\end{align*}
The creation operator $a^*(V)$ is the formal adjoint of $a(V)$, with action given by
 \begin{align*}
(a^*(V) \psi )\uppar n(p, k_1, \dots, k_n) := n^{-1/2} \sum_{j=1}^n v(k_j) \psi \uppar {n-1}(p+k_j,k_1, \dots, \hat{k}_j, \dots, k_n)\, .
\end{align*}
As usual, $\hat{k}_j$ means that the $j$-th variable is omitted. The operator $a^*(V)$ is a densely defined operator on $\hilb$ if and only if $v \in \Lz(\R^d)$. However, in all relevant examples, this is not the case. Often $v$ is in $\Lz_{\mathrm{loc}}(\R^d)$ but is not decaying fast enough at infinity such that $v \notin \Lz$. This is what we will assume in the following. 

If we wanted to start with a renormalisation procedure, we would now simply replace $v$ by $\chi_\Lambda v $ where $\chi_\Lambda$ is the characteristic function of a ball of radius $\Lambda$ in $\R^3$. Instead, we proceed by defining an operator $G_\mu^* := - a(V) L^{-1}_\mu$. Later, we will make assumptions on $v$ which guarantee that this operator is bounded. As a consequence, the symmetric operator $L_{0,\mu} := L_\mu \big \vert_{\ker a(V)}$ is closed for any $\mu \geq 0$. Because $v \notin \Lz$, its domain $\ker a(V)$ is also dense in $\hilb$, see \cite[Lem.~2.2]{nelsontype}. Therefore the adjoint $L_{0,\mu}^*$ is unique. Observe that the operator $G_\mu$ maps elements of $\hilb$ into $\ker L_{0,\mu}^*$, because for all $\psi \in \ker a(V)$ it holds by definition of $G_\mu$ that
\begin{align*}
\langle L_{0,\mu}^* G_\mu \phi,  \psi \rangle =  \langle  \phi, G_\mu^* L_{0,\mu} \psi \rangle = - \langle  \phi, a(V) \psi \rangle = 0 \, .
\end{align*}
We will now define a family of subspaces of the adjoint domain $D(L_{0,\mu}^*)$. In order to do so, we decompose elements of $\hilb$ in the same way as in the theory of point interactions into the sum of two terms: one is regular, i.e. in $D(L)$, and one term is singular, that is, of the form $G_\mu \phi$. If we would like to define a sum of point interaction domains in $\hilb$, we would introduce a boundary or charge space where $\phi$ lives. But because $\hilb$ is an infinite sum, there is another possibility, namely to take $\psi$ itself as the charge. This is what we will do. Note that the decomposition $\psi = (1-G_\mu) \psi + G_\mu \psi$ holds for any $\psi \in \hilb$ and $\mu >0$. Then the family of domains is given by 
\begin{align*}
\dom_\mu := \{ \psi \in \hilb \vert (1-G_\mu) \psi \in D(L) \} \, .
\end{align*}
For $\mu, \lambda >0$, the resolvent identity yields
\begin{align*}
(G_\mu-G_\lambda)^* & =
- a(V)(\lambda-\mu) L_\mu^{-1}  L_\lambda^{-1} 
= ((\lambda-\mu)  L_\mu^{-1} G_\lambda)^{*}  \, .
\end{align*}
In particular it holds that that $1-G_\mu = (1-G_\lambda) - (\lambda-\mu) L_\mu^{-1} G_\lambda$. Because $L_\mu^{-1} G_\lambda$ maps into $ D(L)$, this shows that the domain $\dom_\mu$ is in fact independent of the chosen $\mu>0$. We will denote it by $\dom$ from now on.

In the next step we have to extend the action of $a(V)$ from $D(L)$ to the enlarged domain $\dom$. The formal action of the annihilation operator on the range of $G_\mu$ would read 
\begin{align}
\label{eq:Toddef}
a(V) G_\mu& \psi \uppar n (p, k_1, \dots, k_n) \nonumber
\\
&=  \begin{aligned}[t] &- \psi \uppar n ( p, k_1, \dots, k_{n}) \int_{\R^3}  \frac{\abs{v(k_{n+1})}^2}{L_\mu(p, k_1, \dots, k_{n+1})} \ud k_{n+1} \\
& - \sum_{j=1}^{n} \int_{\R^3} \overline{v(k_{n+1})} v(k_j) \frac{\psi \uppar n ( p + k_j - k_{n+1}, k_1, \dots, \hat{k}_j , \dots, k_{n+1})}{L_\mu(p, k_1, \dots, k_{n+1})} \ud k_{n+1} \, . \end{aligned}
\end{align}
Here $L_\mu(p, k_1, \dots, k_{n+1})$ denotes the functions to which the operator $L_\mu$ reduce to on one sector of $\hilb$ in the Fourier representation. The off-diagonal part of this sum, the second line of \eqref{eq:Toddef}, constitutes an integral operator, which we will denote by $T^\mu_{\uod}$. The integral in the first line of \eqref{eq:Toddef} does in general not converge. In order to regularise this expression, we define the diagonal part of the $T$-operator
\begin{align}
\label{eq:Tddef}
T^\mu_\ud \psi(p, k_1, \dots, k_n) &:= - I_\mu(p,k_1, \dots, k_n) \cdot \psi \uppar n ( p, k_1, \dots, k_{n}) \, ,
\\ \label{eq:integraldef}
\mathrm{where }  \quad  I_\mu(p,k_1, \dots, k_n) &:=  \int_{\R^3}  \frac{\abs{v(k_{n+1})}^2}{L_\mu(p, k_1, \dots, k_{n+1})} - \frac{\abs{v(k_{n+1})}^2}{k_{n+1}^2+\omega(k_{n+1})} \ud k_{n+1} \, .
\end{align}
Now define the action of $T^\mu \psi := T^\mu_{\ud} \psi+T^\mu_{\uod} \psi$ on a (maximal) domain $\mathrm{D}^\mu \subset \hilb$. At first, this definition seems to depend again on the choice of $\mu >0$. Note however that, because the second term of the integral $I_\mu$ in~\eqref{eq:integraldef} is independent of the parameter $\mu >0$, it holds that 
\begin{align}
\label{eq:Tmuindep}
T^\mu - T^\lambda = a(V) (G_\mu-G_\lambda) = a(V) (\lambda-\mu)L^{-1}_\mu G_\lambda = (\mu- \lambda) G_\mu^{*} G_\lambda \, .
\end{align}
Because the operators $G_\mu$ are continuous, this implies that $\psi \in \mathrm{D}^\lambda$ for any $\lambda>0$ as soon as $\psi \in \mathrm{D}^\mu$ for some $\mu>0$. Set $D(T)=\mathrm{D}^\mu $.
While the action of $T^\mu$ does of course still depend on the chosen parameter, this operator gives rise to the desired extension of $a(V)$. We define the action of the full extension for all $\psi \in D(T) \cap \dom$ as
\begin{equation}
\label{eq:defofA}
A^\mu \psi := a(V) (1-G_\mu) \psi +  T^\mu \psi \, .
\end{equation}
 As a consequence of~\eqref{eq:Tmuindep}, we have
\begin{align*}
A^\mu 
&
= a(V) (1-G_\lambda) +  a(V) (G_\lambda-G_\mu) + T^\mu 
= a(V) (1-G_\lambda) +  T^\lambda = A^\lambda \, .
\end{align*}
Therefore we can define the operator $(A,\dom \cap D(T))$ by choosing any $\mu > 0$. Finally we may also define the action of our Hamiltonian manifestly independent of the spectral parameter:
\begin{align*}
H := L^*_{0,0} + A \, .
\end{align*}
Using the definition of $G_\mu$ and $T^\mu$, we can rewrite it in a convenient form that contains the positive spectral parameter:
\begin{align}
\label{eq:symmetricformofH}
H =  (1-G_\mu)^* L_\mu (1-G_\mu) + T^\mu - \mu \, .
\end{align}
In \cite{nelsontype}, it was assumed that $\omega \geq 1$, and as a consequence of the resulting bound $N \leq L$, it was possible to define $G^*:=G_0^*=-a(V)L^{-1}$ without the need for a parameter. We would however like to make clear that the use of a spectral parameter was avoided only for convenience and better readability and is by no means the real benefit of the assumption $\omega \geq 1$.

In order to show self-adjointness of $H$, we will adopt the strategy of~\cite{nelsontype}, where the representation~\eqref{eq:symmetricformofH} (for $\mu=0)$ was used. At first, we have to show that $H_0^\mu := (1-G_\mu)^* L_\mu (1-G_\mu)$ is self-adjoint. In~\cite[Lem.~3.3]{nelsontype} the estimate $N \leq L$ was envoked to show directly the continuous invertibility of $(1-G_0)$, from which the self-adjointness of $H_0^0$ follows. Since we can not use this estimate, we will show that there exists $\mu_0>$ such that $\norm{G_\mu} < 1$ for all $\mu > \mu_0$. The main problem to overcome is however the inclusion $\dom \subset D(T)$ or, more precisely, the relative boundedness of $T^\mu$ with respect to $H_0^\mu$.

The proof of the relative bound for $T^0$ in \cite{nelsontype} makes extensive use of the inequality $N \leq L$ and the resulting fact that $(1-G_0)$ leaves $D(N)$ invariant. For that reason, this strategy is not helpful in the massless case. In fact, because there is no relation between $N$ and $L$, it will be necessary to use characterisations of the domains $D(T)$ and $\dom$ that are independent of $N$ altogether. We will illustrate the problems that occur with this strategy for the example of the Nelson model. While~\cite[Prop.~3.5]{nelsontype} gives -- for this specific model -- an $n$-independent inclusion $D(L^{1/2}) \subset D(T)$, the statement of~\cite[Lem.~3.2]{nelsontype} yields that $G_0$ maps $\hilb$ into $D(L^{\eta})$ for any $0\leq \eta < 1/4$. These exponents do not match together and this is the very problem we have to overcome if we want to define $T^\mu$.
Differentiating between the diagonal and the off-diagonal part of $T^\mu$, we easily observe that, what is actually proven in~\cite{nelsontype} is that on the one hand $D(\Omega^{1/2}) \subset D(T_{\uod})$, but on the other hand $D(L^\eps) \subset D(T_\ud)$ for all $\eps>0$. Thus, at least in the Nelson model, the diagonal part of the operator $T$ seems to pose no problems. The off-diagonal part could be dealt with, if the mapping properties of $G_\mu$ are such that $\dom\subset D(\Omega^{1/2})$. This is exactly what we will prove in the following for a certain class of models under some assumptions on $v$ and $\omega$ in three space dimensions.
\clearpage
\section{Assumptions and Theorems}

Let the dimension of the physical space be equal to three and assume that there exist $\alpha \in [0, 3/2)$ and a constant $c >0$ such that for $v \in \Lz_{\mathrm{loc}}(\R^3)$ it holds that $c (1+\abs{k}^\alpha)^{-1} \leq \abs{v(k)} \leq \abs{k}^{-\alpha} $. Furthermore, there exists $\beta \in (0,2]$ and a constant $\mathsf{m}\geq 0$ such that for $\omega \in L^\infty_{\mathrm{loc}}(\R^3)$ it holds that $ \abs{k}^{\beta} \leq \omega(k) \leq \abs{k}^{\beta}+ \mathsf m$. Defining $D:=1-2 \alpha$ we always assume that $0 \leq D < \beta$.

Note that the Nelson model is contained in this class because $v=\omega^{-1/2}$ allows us to choose $\alpha=1/2$. Clearly $\beta$ is equal to $1$. The upper and lower bounds on $\omega$ hold because $\sqrt{k^2+\mathsf{m}^2} \leq \abs{k}+\mathsf{m}$. It will not be necessary to distinguish between the massive and the massless case, for the only important thing is the pair $(\beta,D)$, which is equal to $(1,0)$ in the Nelson model. Our first result, Proposition~\ref{prop:reverse}, is concerned with regularity properties of a family of domains $\mathfrak{D}^\sigma$. Its proof can be found in Section~\ref{proofofpropreverse}.
\begin{prop}
\label{prop:reverse}
Let $ \beta \in (0,2]$, let $0 \leq D < \beta/2$ if $\beta <2$ and $0<D<1$ if $\beta =2$. Let $\psi\neq 0$ and $\kappa, \eta \in [0,\sigma]$ for some $\sigma \in (0, 1]$. \\
If
\begin{align*}
\psi \in \mathfrak{D}^\sigma = \{ \psi \in \hilb \vert (1-G_\mu) \psi \in D(L^\sigma) \text{ for some } \mu>0 \} \, ,
\end{align*} 
then $\psi \in D(L^\kappa)$ if and only if $\kappa< \frac{2-D}{4}$, and $\psi \in D(\Omega^\eta)$ if and only if $\eta< \frac{2-D}{2 \beta}$.
\end{prop}
Note that also the more general domains $\dom^\sigma$ are independent of the spectral parameter $\mu>0$. If $\dom$ is written without superscript, it is always understood as $\dom^1$. \\ \

To prove self-adjointness of the operator $H$ on $\mathfrak{D}$, we need a more refined condition for the pair of parameters $(\beta,D)$.
\begin{cond}
\label{maincond}
Assume that the pair $(\beta,D)$ satisfies the following inequalities:
\begin{align*}
& 0\leq D< \frac{\beta^2}{2} & \beta \in (0, 2(\sqrt{2}-1)) 
 \\
 &
0\leq D< \frac{2 \beta}{\beta+4}  & \beta \in [2(\sqrt{2}-1),\sqrt{5}-1) 
 \\
 &  0\leq D<\frac{\beta^2 -2 \beta +2   }{\beta+1} & \beta \in [\sqrt{5}-1,2) 
\\
&
0 < D < 2/3 & \beta =2  \, .
\end{align*}
\end{cond}
Theorem~\ref{mainthm} is the main result of this article. It shows, that the only restriction one has to face when extending the construction from massive to massless models is the assumption of the lower bound $D>0$ for $\beta =2$. The upper bound on admissible $D$ is weaker than the one of \cite[Cond~1.1]{nelsontype}, which is  $D < \frac{2 \beta^2}{\beta^2+8}$. Therefore the Theorem~\ref{mainthm} extends the result of the former article to pairs $(\beta , D)$ fulfilling Condition~\ref{maincond}. 

\begin{thm}
\label{mainthm}
If Condition~\ref{maincond} holds, then the operator 
\begin{align*}
H := (L \big \vert_{\ker a(V)})^* + A \, ,
\end{align*}
with $A$ defined in~\eqref{eq:defofA}, is self-adjoint and bounded from below on the domain
\begin{align*}
D(H) := \mathfrak{D} = \{ \psi \in \hilb \vert (1-G_\mu) \psi \in D(L) \text{ for some } \mu>0 \} \, .
\end{align*}
\end{thm}
The proof of Theorem~\ref{mainthm} will be given in Section~\ref{subsectproofsmain}. 
\begin{rem}
The condition $0 \leq D < \beta/2$, which was assumed in Proposition~\ref{proofofpropreverse} \textit{ does not} ensure that $H$ is self-adjoint. However Condition~\ref{maincond} clearly implies that $0 \leq D < \beta/2$, so the statement of Proposition~\ref{prop:reverse} is in particular valid in cases where $\mathfrak{D}=D(H)$ is the domain of the self-adjoint operator $H$ and $\dom^{1/2}$ is its form domain. The Plot~\ref{plot} shows the different regions of admissible pairs of parameters. In general, we consider pairs where $0\leq D < \beta$ for $\beta<2$ and $D \in (0, 2)$ if $\beta=2$. The area below the dotted line, which also excludes the point $(\beta,D)=(2,0)$, is the one for which Proposition~\ref{prop:reverse} charaterises the domain $\dom$. It is, in our language, also the area for which \cite{GrWu17} shows that a renormalisation procedure can be implemented using a Gross transformation. The area below the plain line, again without the point at the right lower corner, is formed by the admissible pairs according to Condition~\ref{maincond}. The area below the dashdotted line is the one that is allowed in \cite[Cond~1.1]{nelsontype}. Because there only massive models are considered, the point $(2,0)$ is however admissible. 
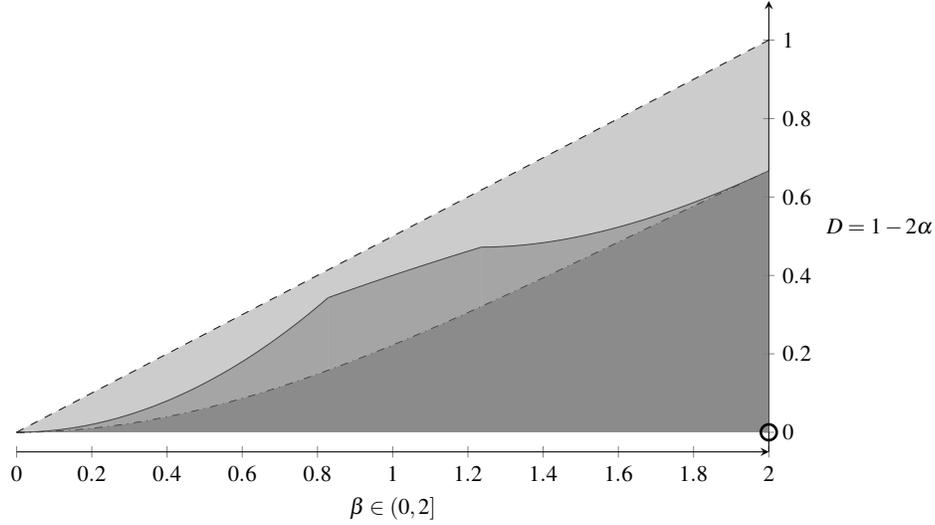
\begin{figure}[H]
  \caption{Admissible Pairs ${(\beta,D)}$}
  \label{plot}
  \begin{tikzpicture}
\begin{axis}[ width=10cm,height=6cm, scale only axis, ymin=-0.05,ymax=1.1,xlabel=${\beta \in (0,2]}$,ylabel = ${D=1-2 \alpha}$, axis x line = bottom,axis y line
=right,xticklabel style={/pgf/number format/precision=5},ylabel style={rotate=270},enlargelimits=false]


\addplot[black,mark size=3pt, line width = 1 pt, only marks,mark=*,fill=blue, fill opacity=0]
coordinates{(2,0.00001)};

\addplot[domain=0:2*(sqrt(2)-1),samples=201,]
{x*x/2};
\addplot[transparent,domain=0:2*(sqrt(2)-1),samples=201,fill=gray,fill opacity=0.5]
{x*x/2} \closedcycle;

\addplot
[domain=2*(sqrt(2)-1):sqrt(5)-1,samples=201]
{2*x/(x+4)};
\addplot
[transparent,domain=2*(sqrt(2)-1):sqrt(5)-1,samples=201,fill=gray,fill opacity=0.5]
{2*x/(x+4)} \closedcycle;

\addplot
[domain=sqrt(5)-1:2,samples=201,]
{(x*x-2*x+2)/(x+1)};
\addplot
[transparent,domain=sqrt(5)-1:2,samples=201,fill=gray,fill opacity=0.5]
{(x*x-2*x+2)/(x+1)} \closedcycle;

\addplot
[
dashdotted,
domain
=0:2,
samples
=201,
]
{(2*x*x)/(x*x+8)};
\addplot
[
transparent,
domain
=0:2,
samples
=201, fill=gray,fill opacity=0.7
]
{(2*x*x)/(x*x+8)} \closedcycle;

\addplot
[dashed, domain=0:2]
{x/2};
\addplot
[transparent, domain=0:2,fill=gray,fill
opacity=0.4]
{x/2} \closedcycle;

\end{axis}
\end{tikzpicture} 
%
%
%
%
%
%
%
%
%
%
\end{figure}
The characterisations of $D(H)$ and $D(\abs{H}^{1/2})$ provide a more detailed answer to Nelson's second question (for the admissible pairs) when compared to the result of Griesemer and Wünsch. First of all, the method in \cite{GrWu17} only allows for the characterisation of the \textit{form domain} of the limiting Hamiltonian. We can reproduce their earlier result here because setting $\sigma=1/2$ in Proposition~\ref{prop:reverse} yields that 
\begin{align*}
D(\abs{H}^{1/2}) = \dom^{1/2} \subset \bigcap_{0 \leq \kappa < \frac{2-D}{4}} D(L^\kappa) \cap D(\Omega^{1/2}) 
\end{align*}
as long as $2-D > \beta$, which is in particular fulfilled for the Nelson model. For determining supersets of the \textit{operator domain} $D(H)=\dom$, the IBC method is the only tool avaliable. For the Nelson model, massive or massless, Proposition~\ref{prop:reverse} implies that $D(H) \subset D(\Omega^\eta)$ for all $\eta <1$ but $D(H) \cap D(\Omega) = \{ 0 \}$.
\end{rem}

\section{Constructing the Hamiltonian}
In the main part of the article we will carry out the program that has been sketched in the introduction. The possibility to set up the operators $G$ and $T$ using positive parameters $\mu>0$ and the results about the \textit{parameter-independence} of the domains $\dom^\sigma$ and the operator $A$ will not be repeated. They can be found in Section~\ref{sec:model}. We will discuss the mapping properties of $G_\mu$ and fit them together with those of $T^\mu$. In this way, we will prove self-adjointness of the Hamiltonian $H$ (Theorem~\ref{mainthm}) and obtain the characterisation of the domains $\dom^\sigma$ in terms of domains of powers of $\Omega$ and $L$ (Proposition~\ref{prop:reverse}).

We will from now on assume that the spectral parameter $\mu$ is greater than one, $\mu\geq 1$. When writing $D(L^x)$ without index for some $x \in \R \setminus \{0 \}$ we mean the domain $D(L_\mu^x)$ for any $\mu \geq 1$. Note also that the assumption on $\mu$ guarantees monotonicity in the exponent, i.e. $L_\mu^x \leq L_\mu^y$ if $x\leq y$.

We will denote by $K$ the collection of variables $K:=(k_1, \dots, k_n)$. Consequently $\hat{K}_j := (k_1, \dots, \hat{k}_j, \dots, k_n)$ is the collection of variables with the $j$-th component omitted. We will use the symbols $L_\mu(p,K)=p^2 + \Omega_\mu(K)$ to denote the functions to which the operators reduce to on one sector of $\hilb$ in the Fourier representation.

Powers of the self-adjoint operators $\Omega$ and $L$ are self-adjoint on their resepctive domains $D(L^\kappa)$ etc., which are all continuously embedded in $\hilb$. We will regard the domains as Banach spaces equipped with the norms $\norm{\psi}_{D(L^\kappa)} = \norm{L^\kappa \psi}_\hilb + \norm{\psi}_\hilb$. The intersection of two such subspaces is a Banach space with norm $\norm{ \psi }_{D(L^\kappa) \cap D(\Omega^\eta)} := \max(\norm{\psi}_{D(L^\kappa)},\norm{ \psi }_{D(\Omega^\eta)})$. We will mostly use the equivalent norm given by the sum, i.e. $\norm{ \psi}_{D(L^\kappa)} + \norm{ \psi}_{D(\Omega^\eta)}$.

\subsection{Mapping Properties of $G_\mu$}
Let us begin with a technical lemma that will be useful later on. It is concerned with certain properties of the affine function $u(s):= (\beta s - D)/2$. This function itself plays an important role in the following because many relations between the parameters can be expressed with its help.
\begin{lem}
\label{lem:thetaass}
Let $ \beta \in (0,2]$, let $0 \leq D < \beta$ if $\beta <2$ and $0<D<2$ if $\beta =2$. Let $\eps_0 >0$ be such that $D+\eps_0 = \beta$. Define for any $0 < \eps < \eps_0$ the function
\begin{align}
\label{eq:thetas}
\theta_\eps(\beta,D) := \begin{cases} \frac{2-D-\eps}{2-\beta} &  \ D > \frac{3 \beta -2}{\beta} - \eps \\
\max(1/{\beta},1) &  \ D \leq \frac{3 \beta -2}{\beta} - \eps 
\, . \end{cases}
\end{align}
Let the affine transformation $u$ for all $s \in [0,\infty)$ be defined as $u(s):= (\beta s - D)/2$. Then it holds that $\theta_\eps \geq 1$. Furthermore $1+u(\theta_\eps) - \theta_\eps \geq \eps$ and $u (\theta_\eps) < 1$.
\end{lem}
\begin{proof}
If $\theta_\eps = 1$, the hypothesis clearly implies that $u(\theta_\eps) <1$. When $\theta_\eps = 1/\beta$, then $u(\theta_\eps) =(1 - D)/2 \leq 1/2 $. If $D > \frac{3 \beta -2}{\beta} - \eps$ then, by definition of $\eps_0$, it holds that $\beta^2 > 3 \beta -2$. This implies that $\beta \in (0,1)$, in particular $\beta/(2-\beta) < 1$ and therefore  $u(\theta_\eps) <  (2-D-\eps-D)/2 < 1$.

In the upper case of \eqref{eq:thetas}, the equality $1+u(\theta_\eps) - \theta_\eps = \eps$ holds by construction. Because $1+u(s)-s$ is non-increasing, it remains to prove that $\frac{2-D-\eps}{2-\beta}$ is an upper bound for $\theta_\eps$. For $1/\beta$ this is the case if and only if $D \leq \frac{3 \beta -2}{\beta} - \eps$. If $\theta_\eps=1$, this follows easily because by definition $2-D-\eps > 2-\beta $.

The last step also proves that $\theta_\eps \geq 1$. \qed
\end{proof}

Now we will consider $G_\mu$ as an operator into $D(L^\kappa)$ under some conditions on $\kappa$. Later, when the target space will be enlarged to $D(\Omega^\eta)$, we will build on some of the formulas obtained here.
\begin{lem}
\label{lem:GintoL}
Let $ \beta \in (0,2]$, let $0 \leq D < \beta$ if $\beta <2$ and $0<D<2$ if $\beta =2$. Then for any $0\leq \kappa < (2-D)/4$ and any $\mu\geq 1$ it holds that $G_\mu$ is continuous from $D(\Omega^{\kappa})$ to $D(L^\kappa)$. There exists $\mu_0 \geq 1$ such that the norm of $G_\mu$ is smaller than $1$ for all $\mu > \mu_0$.
\end{lem}
\begin{proof}
We will show that $\norm{L^\kappa G_\mu \psi} \leq C \norm{\Omega_\mu^{\kappa-(1+u(s)-s)} \psi}$ for some constant $C>0$ and any $s \geq 1$. In view of Lemma~\ref{lem:thetaass}, this proves the claim because 
\begin{align*}
 \norm{\Omega_\mu^{\kappa-\eps/2} \psi} \leq \mu^{-\eps/2} \norm{\Omega_\mu^{\kappa} \psi}  \leq \mu^{-\eps/2} \norm{ \psi}_{D(\Omega_\mu^{\kappa})} \, .
 \end{align*}
For later use, we will write $\Xi_\mu := L_\mu$ at first. To estimate $\abs{\Xi_\mu^\eta G_\mu \psi}^2$, we multiply by $\omega(k_j)^s/\omega(k_j)^s$ for $s \geq 1$ and use the finite dimensional Cauchy-Schwarz inequality:
  \begin{align*}
\abs{\Xi_\mu^\kappa G_\mu \psi \uppar n (p,K)}^2
&
\leq
 \sum_{j=1}^{n+1}  \sum_{\nu=1}^{n+1} \frac{\omega(k_\nu)^s }{n+1} \frac{\abs{ v(k_j)}^2 \Xi_\mu(p,K)^{2 \kappa} \abs{\psi \uppar n(p+ k_j,\hat{K}_j)}^2}{L_\mu(p,K)^{2} \omega(k_j)^s } 
 \\
 &
 \leq
 \sum_{j=1}^{n+1}  \frac{\omega(k_j)^s + \Omega(\hat{K}_j)^s}{n+1} \frac{\abs{ v(k_j)}^2 \Xi_\mu(p,K)^{2 \kappa} \abs{\psi \uppar n(p+ k_j,\hat{K}_j)}^2}{L_\lambda(p,K)^{2} \omega(k_j)^s } 
   \, .
 \end{align*}
 In the second step, the fact that $s \geq 1$ is essential. We now use the assumptions $\abs{v(k)} \leq  \abs{k}^{-\alpha}$ and $\omega(k) \geq \abs{k}^{\beta}$. This yields for the translated expression $\abs{\Xi_\mu^\kappa G_\mu\psi \uppar n (p- k_j,K)}^2$ the bound 
  \begin{align*}
\abs{\Xi_\mu^\kappa G_\mu\psi \uppar n (p- k_j,K)}^2  \leq & \sum_{j=1}^{n+1}  \frac{\abs{\psi \uppar n(p,\hat{K}_j)}^2}{ n+1} \frac{ \Xi_\mu(p-k_j,K)^{2 \kappa} \abs{k_j}^{-2\alpha-\beta s} \Omega(\hat{K}_j)^s}{ {L_\mu(p- k_j,K)^{2} }}  
\\
&
+
\sum_{j=1}^{n+1}  \frac{\abs{\psi \uppar n(p,\hat{K}_j)}^2}{n+1}   \frac{\Xi_\mu(p-k_j,K)^{2 \kappa} \abs{k_j}^{-2\alpha}} {L_\mu(p- k_j,K)^{2} }  \, .
 \end{align*}
Now we use the symmetry of $\psi$, $L$ and $\Xi$ to note that we can bound the integral over these sums by the integral over the first term of the sums times $n+1$. That is, we have a bound
\begin{align*}
\norm{\Xi_\mu^\kappa G_\mu \psi \uppar n}^2 &= \int \abs{\Xi_\mu^\kappa G_\mu\psi \uppar n (p- k_j,K)}^2 \ud K \ud p
\\
& \leq \int \abs{\psi \uppar n(p,\hat{K}_1)}^2 \int_{\R^3} \gamma_\ud^{\Xi_\mu} +  \gamma_\uod^{\Xi_\mu} \ \ud k_1 \ \ud \hat{K}_1 \ud p
\end{align*}
where
\begin{align}
\label{eq:defofgammas}
\gamma_\ud^{\Xi_\mu}(p,K) + \gamma_\uod^{\Xi_\mu}(p,K) :=  \frac{\Xi_\mu(p-k_1,K)^{2 \kappa} } {L_\mu(p- k_1,K)^{2} \abs{k_1}^{2\alpha}}  + \frac{ \Xi_\mu(p-k_1,K)^{2 \kappa}  \Omega(\hat{K}_1)^s}{ {L_\mu(p- k_1,K)^{2} }\abs{k_1}^{2\alpha+\beta s}}  .
\end{align}

We now specify to $\Xi_\mu=L_\mu$ and estimate it from below by $\abs{p-k_1}^2+\Omega_\mu(\hat{K}_1)$. Recall that since $D\geq 0$ we have by hypothesis $\kappa < 1/2$. So we can bound the integral over $k_1$ of the off-diagonal part by
  \begin{align*}
\int_{\R^3} \gamma_\uod^{L_\mu}(p,K) \, \ud k_1 \leq   \int_{\R^3}  \frac{ \Omega_\mu(\hat{K}_{1})^s  \abs{k_1}^{-2\alpha-\beta s}}{  (\abs{p-k_1}^2+\Omega_\mu(\hat{K}_1))^{2(1-\kappa)}  }   \ud k_1 \, .
 \end{align*}
If $u(s)<1$ and $2 \kappa < u(s) +1$, this integral is by scaling bounded by a constant times
 \begin{align*}
 \Omega_\mu(\hat{K}_1)^{s+2(\kappa-1)+\frac{3-2\alpha-\beta s}{2}} =  \Omega_\mu(\hat{K}_1)^{2 \left(\kappa-\frac{1+u(s)-s}{2} \right)} \, .
 \end{align*}
If $2 \kappa < u(0) +1$, we obtain similarly for some $C>0$ a bound for the diagonal part: 
 \begin{align*}
\int_{\R^3} \gamma_\uod^{L_\mu}(p,K) \, \ud k_1 \leq C \Omega_\mu(\hat{K}_1)^{2(\kappa-1)+\frac{3-2\alpha}{2}} = C \Omega_\mu(\hat{K}_1)^{2 \left(\kappa-\frac{1+u(0)}{2} \right)} \, .
 \end{align*}
Because $\beta>0$, the function $u$ is increasing so the hypothesis $2 \kappa < u(0) +1 = (2-D)/2$ clearly implies $2 \kappa < u(s) +1$. In addition $\beta \leq 2$, so we can estimate $1+u(s)-s \leq 1+u(0)$. \qed
\end{proof}

The next lemma deals with the most important step of the construction, namely the mapping properties of $G_\mu$ into $D(\Omega^\eta)$. It is only here (because more explicit computations are used) where the fact that the dimension is equal to three is relevant.

\begin{lem}
\label{lem:GintoOmega}
Let $ \beta \in (0,2]$, let $0 \leq D < \beta$ if $\beta <2$ and $0<D<2$ if $\beta =2$. Assume that there exists $\eps_\uod >0$ small enough such that
\begin{align*}
0 \leq  
 \eta \leq  \begin{cases}  \frac{\frac{\beta(2-D-\eps_\uod)}{2-\beta} - D +1}{2\beta} & \ D > \frac{3 \beta -2}{\beta} - \eps_\uod ,  \\
\frac{2- D}{2 \beta} & \ D \leq \frac{3 \beta -2}{\beta} - \eps_\uod \, . \end{cases} 
\end{align*}
Define for any $\eps_\ud \geq 0$ the map $q_{\eps_\ud}(\eta):=\max\left(0,\eta+\eps_\ud  - (\beta +2 -2 D) /(4 \beta)\right)$. Then for any $\mu\geq 1$ and any $\eps_\ud >0 $ it holds that $G_\mu$ is continuous from $D(\Omega^{\eta}) \cap D(L^{q_{\eps_\ud}(\eta)})$ to $D(\Omega^\eta)$ and there exists $\mu_0\geq 1$ such that the norm of $G_\mu$ as a map between these two spaces is smaller than $1$ for all $\mu > \mu_0$.
\end{lem}
\begin{proof}
To estimate the norm of $\Omega_\mu^\eta G_\mu \psi \uppar n$, we start directly with the expressions $\gamma_\ud^{\Omega^\eta}$ and $\gamma_\uod^{\Omega^\eta}$ as they have been defined in \eqref{eq:defofgammas}. Note that we have replaced the exponent $\kappa$ by $\eta$. By defining the rescaled variables $\tilde p := p/\Omega_\lambda^{1/2}$ and $ \tilde k := k_1 /\Omega_\lambda^{1/2}$ we can estimate
\begin{align*}
 \int_{\R^3} \gamma_\ud^{\Omega^\eta}(p,K) \, \ud k_1 
&
\leq  \int_{\R^3}  \frac{\left(\abs{k_1}^\beta + \mathsf{m} + \Omega_\mu(\hat{K}_1) \right)^{2 \eta}  \abs{k_1}^{-2\alpha}}{  (\abs{p-k_1}^2+   \abs{k_1}^\beta+\Omega_\mu(\hat{K}_1))^{2}  } \,  \ud k_1 \\
&
=
\Omega_\mu(\hat{K}_{1})^{2 \eta  - (u(0)+1)} 
 \int_{\R^3}  \frac{\left(\abs{\tilde k}^\beta  + \mathsf{m} + 1 \right)^{2 \eta}  \abs{\tilde k}^{-2\alpha}}{  \left(\abs{\tilde p - \tilde k}^2+   \Omega_\mu(\hat{K}_1)^{\frac{\beta-2}{2}} \abs{\tilde k}^\beta +1 \right)^{2}  }   \ud \tilde k \, .
\end{align*}
In the very same way we obtain for the integral over $k_1$ of the off-diagonal part in \eqref{eq:defofgammas} the upper bound 
\begin{align*}
\Omega_\mu(\hat{K}_{1})^{2 \eta  - (1+u(s)-s)}  \int_{\R^3}  \frac{\left(\abs{\tilde k}^\beta  + \mathsf{m} + 1 \right)^{2 \eta}  \abs{\tilde k}^{-2\alpha-\beta s}}{  \left(\abs{\tilde p - \tilde k}^2+   \Omega_\mu(\hat{K}_1)^{\frac{\beta-2}{2}} \abs{\tilde k}^\beta +1 \right)^{2}  }   \ud \tilde k \, .
\end{align*}
Abbreviate $\Omega:=\Omega_\mu(\hat{K}_1)$, set $\mathsf{M} := \mathsf{m} + 1 \in (0,\infty)$ and denote the remaining integral by
\begin{align}
\label{eq:upsilon}
\Upsilon(s,\mu,\tilde p):=\int_{\R^3}  \frac{\left(\abs{\tilde k}^\beta  + \mathsf{M} \right)^{2 \eta}  \abs{\tilde k}^{-2\alpha-\beta s}}{  \left(\abs{\tilde p - \tilde k}^2+   \Omega^{\beta/2-1} \abs{\tilde k}^\beta +1 \right)^{2}  }   \ud \tilde k \, .
\end{align}
 
The integral $\Upsilon$ is clearly bounded for any $\tilde{p} \in \R^3$ as long as $\eta < \frac{1+u(s)}{\beta}$ and $u(s)<1$. If $\abs{ \tilde{p}} \leq 1$, we therefore estimate it simply by a constant. So assume in the following that $\abs{ \tilde{p}} > 1$ and compute using spherical coordinates 
  \begin{align*}
\Upsilon&(s,\mu,\tilde p)
= 2 \pi \int_0^\infty \int_{-1}^1 \frac{(r^\beta+\mathsf{M})^{2\eta} r^{2-2\alpha-\beta s}}{ (r^2 + \tilde{p}^2 -2 r \tilde{p} \sigma +r^\beta \Omega^{\frac{\beta-2}{2}} +1)^2} \, \ud r \ud \sigma
\\
&
= 2 \pi \int_0^\infty  \frac{(r^\beta+\mathsf{M})^{2\eta} r^{2-2\alpha-\beta s}}{ ((r- \tilde{p})^2 + r^\beta \Omega^{\frac{\beta-2}{2}}+1)((r+ \tilde{p})^2  +r^\beta \Omega^{\frac{\beta-2}{2}} +1)} \, \ud r 
\\
&
\leq 
2 \pi  (\tilde p^2)^{ \eta \beta -(u(s)+1)} \int_0^\infty  \frac{ ( x^\beta +\mathsf{M})^{2\eta} + x^{2-2\alpha-\beta s}}{ ((x- 1)^2 x^\beta p^{\beta-2}   +\tilde{p}^{-2})((x+ 1)^2 + x^\beta p^{\beta-2}   +\tilde{p}^{-2} )} \ud  x \, .
 \end{align*} 
We have replaced $\mathsf{M}/\tilde{p}^\beta$ simply by $\mathsf{M}$ because $\abs{ \tilde{p}} > 1$. The integral from $x=2$ to infinity is bounded by a constant, independent of $\tilde p $, for any $\eta < \frac{1+u(s)}{\beta}$. The same is true of the integral from zero to $x=2^{-1/\beta}$. Consider the integral from $2^{1/\beta}<1$ to $2$. On this interval, the numerator of the integral can be estimated by a constant that depends on $\mathsf{M}$, the factor in the denominator that contains the $(x+1)^2$-term is bounded from below by one. It remains to estimate the factor which has a pole at $x=1$. This can be done by enlarging the domain and making use of fact that the antiderivative of $(1+x^2)^{-1}$ is the $\arctan$. So we have
   \begin{align*}
\int_{2^{-1/\beta}}^2 &\frac{1}{ ((x- 1)^2 + x^\beta p^{\beta-2}  +\tilde{p}^{-2})} \ud  x  
\leq \int_{2^{-1/\beta}}^2 \frac{1}{ ((x- 1)^2 + 1/2  {p}^{\beta-2}  +\tilde{p}^{-2})} \ud  x  
\nonumber \\
&
\leq  \int_{\R} \frac{1}{ ((x- 1)^2 +1/2  {p}^{\beta-2}   +\tilde{p}^{-2})} \ud  x  
= \pi \left[1/2  {p}^{\beta-2}   +\tilde{p}^{-2}  \right]^{-1/2} 
  \, .
 \end{align*} 
Recall that the other parts of this integral are bounded by a constant. So, because $\tilde{p} >1$ implies ${p} >1$, we can bound as a whole:
\begin{align}
\label{eq:boundwitht}
\chi_{\{\tilde{p} >1\}} \, \Upsilon(s,\mu,\tilde p) 
&
\leq \chi_{\{\tilde{p} >1\}}\left( C +  \left[1/2  {p}^{\beta-2}   +\tilde{p}^{-2}  \right]^{-1/2}\right)
 \nonumber \\
&
\leq C' \chi_{\{\tilde{p} >1\}} ( {p}^{\frac{2-\beta}{2}})^{(1-t)} (\tilde{p})^{t}  \, .
\end{align}
Here we have introduced a parameter $t \in [0,1]$. 
Now we have to distinguish between the diagonal term in~\eqref{eq:defofgammas}, where we have $s=0$ and choose $t=0$ in~\eqref{eq:boundwitht}, and the off-diagonal term  where we choose $t=1$ in~\eqref{eq:boundwitht} and observe that $s \geq 1$ is required. The off-diagonal term hence can be bounded by
\begin{align*}
 \int_{\R^3} \gamma_\uod^{\Omega^\eta}(p,K) \, \ud k_1 \leq C \Omega_\mu(\hat{K}_{1})^{2 \eta  - (u(s)+1-s)} \left( \chi_{\{\tilde p \leq 1 \}} +  \chi_{\{\tilde p > 1\}}   \tilde p^{2 \eta \beta -2(u(s)+1))+1}\right)   \, .
\end{align*}
We would like to have -- for the off-diagonal term -- a bound independent of $p$. To achieve this, we apply Lemma~\ref{lem:thetaass} and choose $s=\theta_{\eps_\uod}$ for an ${\eps_\uod}>0$ admissible there. Then we can see that our upper bounds on $\eta$ are such that the exponent of $\tilde p$ is non-positive. This is because for $s=\theta_{\eps_\uod}$ the exponent becomes
\begin{align*}
2 \eta \beta -2(u(\theta_{\eps_\uod})+1))+1 
&
= 2 \beta  \eta - 2 \beta \begin{cases}  \frac{\frac{\beta(2-D)-\beta {\eps_\uod}}{2-\beta} - D +1}{2\beta} &  \ D > \frac{3 \beta -2}{\beta} - {\eps_\uod} \\
\frac{\beta \max(1,1/\beta)- D+1}{2 \beta} &  \ D \leq \frac{3 \beta -2}{\beta} -{\eps_\uod} \, , \end{cases}
\end{align*}
and obviously $1 \leq \beta \max(1,1/\beta)$. These considerations imply that the norm of the off-diagonal term is bounded by $\norm{ \Omega_\mu^{\eta-{\eps_\uod}/2} \psi}^2 \leq \mu^{-{\eps_\uod}} \norm{ \Omega_\mu^{\eta} \psi}$.

We are not able to obtain a bound independent of $p$ also for the diagonal term in \eqref{eq:defofgammas}. Setting $s=t=0$ in \eqref{eq:boundwitht}, yields for the integral $\int_{\R^3} \gamma_\ud^{\Omega^\eta}(p,K) \, \ud k_1$ a bound of the form constant times
\begin{align*}
 &\Omega_\mu(\hat{K}_{1})^{2 \eta  - (u(0)+1)} \left( \chi_{\{\tilde p \leq 1 \}} +  \chi_{\{\tilde p > 1\}}  \Omega_\mu(\hat{K}_1)^{- \eta \beta +(u(0)+1)}  p^{2 \eta \beta -2(u(0)+1)+\frac{2-\beta}{2}} \right) 
 \\
 &
=
\Omega_\mu(\hat{K}_{1})^{2 \eta  - (u(0)+1)} \chi_{\{\tilde p \leq 1 \}}  +  \chi_{\{\tilde p > 1\}} \Omega_\mu(\hat{K}_{1})^{2 \eta  -  \eta \beta }  p^{2 \beta \left(\eta - (\beta +2 -2 D) /(4 \beta)\right)}
 \, .
\end{align*}
Due to the fact that $D<\beta \leq 2$, the first term here is bounded by $\mu^{-u(0)-1} \Omega_\mu^{2 \eta}$ for all $\tilde p \in [0, \infty)$. To bound the second term, introduce an $\eps_\ud > 0$, which yields
\begin{align*}
 \chi_{\{\tilde p > 1\}} \Omega_\mu & (\hat{K}_{1})^{2 \eta  -  \eta \beta }  p^{2 \beta \left(\eta - (\beta +2 -2 D) /(4 \beta)\right)}
 \\
 &
 \leq \chi_{\{\tilde p > 1\}} \Omega_\mu  (\hat{K}_{1})^{\eta (2 -  \beta)  } \mu^{-\eps_\ud \beta}  (p^2+\mu)^{\beta \left(\eta + \eps_\ud - (\beta +2 -2 D) /(4 \beta)\right) }
 \\
 &
  \leq  \Omega_\mu  (\hat{K}_{1})^{ \eta (2 -  \beta) } \mu^{-\eps_\ud \beta}  (p^2+\mu)^{{\beta} q_{\eps_\ud}(\eta)} 
  \end{align*}
  We have used in particular that $\mu \geq 1$ to get rid of the characteristic function. Now we apply Young's inequality with $\nu=2/(2-\beta)$ and $\xi = 2/\beta$, which leads to the upper bound
  \begin{align*}
 C \mu^{-\eps \beta} \left( \Omega_\mu(\hat{K}_{1})^{2 \eta   } +  (p^2+\mu)^{ 2 q_{\eps_\ud}(\eta)} \right)
 \, .
\end{align*}
 Because $\beta >0$, the norm of this term goes to zero as $\mu \rightarrow \infty$. This proves the claim. \qed
\end{proof}

The Neumann series is a candidate for the inverse of the operator $1-G_\mu$. On domains where the norm of $G_\mu$ is decreasing, the series will converge for large enough $\mu$.
\begin{cor}
\label{cor:1minusG}
Let $ \beta \in (0,2]$, let $0 \leq D < \beta$ if $\beta <2$ and $0<D<2$ if $\beta =2$. Let $\eta, \kappa \geq 0$. Assume that for any $\eps >0$ small enough
\begin{align}
\label{eq:completacond}
0 \leq  
 \eta <  \begin{cases}  \frac{\frac{\beta(2-D-\eps)}{2-\beta} - D +1}{2\beta} & \ D > \frac{3 \beta -2}{\beta} - \eps ,  \\
\frac{2- D}{2 \beta} & \ D \leq \frac{3 \beta -2}{\beta} - \eps  \end{cases} 
\end{align}
and $\max(\kappa,q_0(\eta)) < \frac{2-D}{4}$. Then there exists $\mu_0 \geq 1$ such that $1-G_\mu$ is continuously invertible on $D(\Omega^\eta)\cap D(L^{\max(\kappa,q_\eps(\eta))})$ for any $\mu>\mu_0$, possibly for a smaller $\eps>0$.
 \end{cor}
 \begin{proof}
We make $\eps>0$ possibly smaller, such that also $\max(\kappa,q_\eps(\eta)) < \frac{2-D}{4}$. Then Lemma~\ref{lem:GintoL} implies that for any $\eta\geq 0$ it holds that
\begin{align*}
\norm{G_\mu \psi}_{D(L^{\max(\kappa,q_\eps(\eta))})} \leq c(\mu) \norm{\psi}_{D(\Omega^{\max(\kappa,q_\eps(\eta))})} \leq c(\mu) \norm{\psi}_{D(\Omega^\eta)\cap D(L^{\max(\kappa,q_\eps(\eta))})} 
\end{align*}
with $c(\mu) <1$ for $\mu$ larger than some $\mu_0 \geq 1$. Due to the assumptions we have made on $\eta$, the Lemma~\ref{lem:GintoOmega} gives
\begin{align*}
\norm{G_\mu \psi}_{D(\Omega^\eta)} \leq C(\mu) \norm{\psi}_{D(\Omega^\eta)\cap D(L^{q_\eps(\eta)})} \leq C(\mu) \norm{\psi}_{D(\Omega^\eta)\cap D(L^{\max(\kappa,q_\eps(\eta))})} 
\end{align*}
with $C(\mu)<1$ if $\mu > \mu_0$ for some $\mu_0 \geq 1$. The last inequality simply holds because $\mu \geq 1$ and $q_\eps(\eta)\leq \max(\kappa,q_\eps(\eta))$. \qed
 \end{proof}
 
 We are now ready to prove that the "free`` operator $H_0^\mu:=(1-G_\mu)^* L_\mu (1-G_\mu)$ is self-adjoint. To prove self-adjointness of the whole operator $H$ in Section~\ref{subsectproofsmain}, the operator $T^\mu$ will be regarded as an operator perturbation of $H_0^\mu$.
\begin{cor}
\label{cor:Hnull}
Let $ \beta \in (0,2]$, let $0 \leq D < \beta$ if $\beta <2$ and $0<D<2$ if $\beta =2$. Then $H_0^\mu$ is self-adjoint and positive on $D(H_0^\mu)= \mathfrak{D}=\{ \psi \in \hilb \vert (1-G_\mu) \psi \in D(L) \text{ for some } \mu>0 \}$.
\end{cor}
\begin{proof}
Apply Corollary~\ref{cor:1minusG} with $\eta=\kappa=0$. This is possible because the upper bounds on $\eta$ and $\kappa$ are positive for $D<\beta$ and in addition $q_0(0) \leq 0$. That means that $(1-G_\mu)$ is invertible on $\hilb$ for $\mu\geq 1$ large enough, so $D(H_0^\mu):= \mathfrak{D}$ is dense in $\hilb$. The operator $H_0^\mu$ is clearly symmetric and positive and it is easy to see that $\phi \in D((H_0^\mu)^*)$ implies $\phi \in D(H_0^\mu)$. \qed
\end{proof}

 \subsection{The Domain $\mathfrak{D}$: Proof of Proposition~\ref{prop:reverse}}
 \label{proofofpropreverse}
In order to determine supersets for $\mathfrak{D}$, we can now build on the results of the previous section. The domain can be characterised as $\mathfrak{D} = (1-G_\mu)^{-1} D(L)$ for any $\mu\geq 1$ admissible in Corollary~\ref{cor:Hnull}. Therefore any subspace of the form $(1-G_\mu)^{-1} \mathcal{S}$ with $D(L) \hookrightarrow \mathcal{S} \subset \hilb$ is also a superset for $\mathfrak{D}$. If $1-G_\mu$ is invertible on $(\mathcal{S},\norm {\cdot}_\mathcal{S})$, we have  $(1-G_\mu)^{-1} \mathcal{S} = \mathcal{S}$, which then allows us to explicitly characterise this space. In this section, we will restrict the range of parameters to pairs where $D < \beta /2$ in contrast to $\beta$. In this way, the various conditions on $\eta$ can be significantly simplified.

 \begin{prop}
 \label{prop:infibounds}
Let $ \beta \in (0,2]$, let $0 \leq D < \beta/2$ if $\beta <2$ and $0<D<1$ if $\beta =2$. Define for any $\sigma \in (0,1]$ the subspace $\mathfrak{D}^\sigma=\{ \psi \in \hilb \vert (1-G_\mu) \psi \in D(L^\sigma) \text{ for some } \mu>0 \}$.
\begin{itemize}
\item For any $\eta \in [0,\sigma]$ with $\eta < \frac{2-D}{2 \beta}$ it holds that $\mathfrak{D}^\sigma \subset D(\Omega^\eta) \cap D(L^{q_\eps(\eta)})$ for any $\eps>0$ small enough.
\item
For any $\kappa \in [0,\sigma]$ with $\kappa < \frac{2-D}{4}$ it holds that $\mathfrak{D}^\sigma \subset  D(L^{\kappa})$.
\end{itemize}
 \end{prop}
 \begin{proof}
The first task will be to perform the promised simplification of the conditions on $\eta$ in Corollary~\ref{cor:1minusG}.
First, observe that $\eta \leq \sigma$ means of course also $\eta \leq 1$. We will now prove that $\eta \leq 1$ together with $D<\beta/2$ implies that, if $\eps >0$ can be arbitrarily small, then
\begin{align*}
 \eta<   \frac{\frac{\beta(2-D-\eps)}{2-\beta} - D +1}{2\beta} \quad \text{ if } \quad D > \frac{3 \beta -2}{\beta} - \eps  \, .
\end{align*}
To show this, observe that $\frac{3 \beta -2}{\beta}<D+\eps<\beta/2+\eps$ means that $\beta$ has to fulfill the inequality $2 \beta \eps > 6 \beta - \beta^2 -4 $. This can, for $\eps$ small enough, only be satisfied for $\beta < 4/5$. Using again $D<\beta/2$ we bound, possibly making $\eps>0$ smaller,
\begin{align*}
\frac{\frac{\beta(2-D-\eps)}{2-\beta} - D +1}{2\beta} -1  &
> \frac{\beta(2-\beta/2-\eps) - (\beta/2 -1+2\beta)(2-\beta)}{2\beta (2-\beta)} 
\\
&
=
\frac{(1-\beta)^2-\beta \eps}{\beta (2-\beta)} > \frac{5^{-2}-(4/5) \eps}{2 (4/5)} >0 \, .
\end{align*}
To sum up, we have shown that if $\eta \leq 1$ then the upper case of \eqref{eq:completacond} is fulfilled. The lower case in this very condition is also satisfied by hypothesis.

Our second step is to show that the assumptions $\eta \leq 1$ and $\eta < \frac{2-D}{2 \beta}$ are such that also $q_0(\eta) < \frac{2-D}{4}$. Note that the latter condition is equivalent to $ \eta < \frac{2-D}{4}+\frac{\beta+2(1-D)}{4 \beta}$. Using $D<\beta/2$ we now bound from below
 \begin{align*}
  \frac{2-D}{4}+\frac{\beta+2(1-D)}{4 \beta}-1 > \frac{4-4 \beta -\beta^2}{8 \beta} \\
 \text{and \ } \ \frac{2-D}{4}+\frac{\beta+2(1-D)}{4 \beta}-\frac{2-D}{2 \beta } > \frac{6 \beta -\beta^2-4}{8 \beta} \, .
 \end{align*}
Observe that for any $\beta$ at least one of these functions is positive. So if either $\eta \leq 1$ or $\eta < \frac{2-D}{2 \beta }$ then also $q_0(\eta) < \frac{2-D}{4}$. The above considerations allow us to apply the Corollary~\ref{cor:1minusG} and proceed with the main part of the proof.
 
For $\eta, \kappa$ fulfilling the hypothesis, we define $S_1 := \Omega^\eta$ and $S_2 := L^{\max(\kappa,q_\eps(\eta))}$ and $\mathcal{S} = (D(S_1)\cap D(S_2), \norm{ \cdot}_{D(S_1)} + \norm{ \cdot}_{D(S_2)})$. Recall that $\mu \geq 1$ implies $D(L^\sigma) \hookrightarrow D(L^{\min(\eta,\kappa)})$. Therefore we may consider the chain of inclusions $D(L) \hookrightarrow D(L^\sigma) \hookrightarrow \mathcal{S}$.  Furthermore $\norm{S_i \psi}_{\hilb} \leq \norm{\psi}_\mathcal{S}$ and denoting $C_\mu := \norm{(1-G_\mu)^{-1} }_{\mathcal{L}(\mathcal{S})}$ we have
 \begin{align}
&\norm{S_i \psi}_{\hilb}  \leq \norm{\psi}_{\mathcal{S}} = \norm{(1-G_\mu)^{-1} (1-G_\mu) \psi}_{\mathcal{S}} 
\leq C_\mu \norm{(1-G_\mu) \psi}_{\mathcal{S}}
 \label{eq:relbound} \\ \nonumber 
&
\leq C_\mu C' \norm{(1-G_\mu) \psi}_{D(L^\sigma)} 
= C_\mu C' \left( \norm{L^\sigma_\mu (1-G_\mu) \psi}_{\hilb} +  \norm{ (1-G_\mu) \psi}_{\hilb} \right) \, .
\end{align}
Inserting $1=(1-G_\mu)^{-*} (1-G_\mu)^*$ yields the desired bound. In order to obtain the first part of the statement, we can set $\kappa=0$. For the second part we choose $\eta=0$ which implies $D(\Omega^\eta) = \hilb$ and $q_\eps(\eta)\leq 0$ for $\eps>0$ small enough. \qed
 \end{proof}
 
 \begin{cor}
 \label{cor:infibounds2}
 Let $ \beta \in (0,2]$, let $0 \leq D < \beta/2$ if $\beta <2$ and $0<D<1$ if $\beta =2$.
 \begin{itemize}
\item For any $\eta \in [0,1)$ with $\eta < \frac{2-D}{2 \beta}$ there exists $\mu_0 \geq 1$ such that for any $\mu > \mu_0$ the operator $\Omega_\mu^\eta$ is infinitesimally bounded with respect to $H_0^\mu$ 
\item
For any $\kappa \geq 0$ with $\kappa < \frac{2-D}{4}$ there exists $\lambda_0 \geq 1$ such that for any $\lambda > \lambda_0$ the operator $L_\lambda^\kappa$ is infinitesimally bounded with respect to $H_0^\lambda$.
\end{itemize}
 \end{cor}
 \begin{proof}
Because $\eta<1$, by Young's inequality, we have
\begin{align}
\label{eq:Young}
\norm{L_\mu^\eta \phi}  \leq \tilde C ( \eps \norm{L_\mu \phi} + \eps^{- \eta/(1-\eta)} \norm{ \phi}) \, .
\end{align}  
for any $\eps>0$ and any $\phi \in D(L)$. In~\eqref{eq:relbound} we can set $\sigma=\eta$, and because $\phi = (1-G_\mu) \psi \in (1-G_\mu) \mathfrak{D}^\eta \subset D(L)$, we can use~\eqref{eq:Young} such that
 \begin{align*}
\norm{\Omega^\eta \psi}_{\hilb} & \leq C_\mu C' \tilde C  \left(\eps \norm{L_\mu (1-G_\mu) \psi}_{\hilb} + (1+\eps^{- \eta/(1-\eta)}) \norm{ (1-G_\mu) \psi}_{\hilb} \right) \, .
\end{align*}
Using $1=(1-G_\mu)^{-*} (1-G_\mu)^*$, we prove infintesimal boundedness of $\Omega^\eta$ with respect to $H_0^\mu$ if $\mu$ is large enough. The case of $L^\kappa$ can be proved in exactly the same way.
 \qed
 \end{proof}
%

Now we are well prepared to prove Proposition~\ref{prop:reverse}.

\begin{proof}[\textbf{Proof of Proposition~\ref{prop:reverse}}]{~}
One of the implications is provided by Proposition \ref{prop:infibounds}. It remains to prove that $0 \neq \psi \in \mathfrak{D}^\sigma$ implies that $\norm{L^\kappa \psi}$ or $\norm{\Omega^\eta \psi}$ are infinite if $\kappa \geq \frac{2-D}{4}$ or $\eta \geq \frac{2-D}{2 \beta}$, respectively. For later use we write $\Xi_\mu$ to denote either $L_\mu$ or $\Omega_\mu$. Decomposing $\Xi^\kappa \psi =\Xi^\kappa_\mu (1-G_\mu) \psi +\Xi^\kappa_\mu G_\mu \psi$ we see that, because in any case $\kappa,\eta\leq \sigma$, the norm of the first term is always finite. Recall that we have $\mu \geq 1$. Choose $n \in \N$ such that $\psi \uppar n \neq 0$. For any $r >0$ we define the set
\begin{align*}
U_r := \{(p,K) \in \R^{3+3(n+1)} \vert \abs{p}<r, \ \abs{k_j}< r \text { for all }  2 \leq j \leq {n+1} \} \, .
\end{align*}
We will now show that we can choose $r>0$ such that $\norm{\Xi^\kappa_\mu G_\mu \psi\uppar n}_{\Lz(U_r)}^2$ is infinite. To do so we will split the sum that consitutes $G_\mu$ and apply the inequality $\frac{t-1}{t} a^2 -(t-1) b^2 \leq \abs{a+b}^2$ for $t = 2$. In addition we use that $(\sum_{j=1}^n a_j)^2 \leq n \sum_{j=1}^n a_j^2$. Taken together, this leads to the lower bound
\begin{align}
 \abs{L_\mu G_\mu \psi \uppar n}^2 \geq 
\begin{aligned}[t]
& \frac{\abs{ v(k_1)}^2 \Xi_\mu(p,K)^{2 \kappa}\abs{\psi(p+k_{1},\hat{K}_{1})}^2}{2 (n+1) L_\mu(p,K)^2} 
\\
&-  \sum_{j=2}^{n+1} \frac{\abs{ v(k_j)}^2 \Xi_\mu(p,K)^{2 \kappa} \abs{\psi(p+k_{j},\hat{K}_{j})}^2}{L_\mu(p,K)^2}  \, . \end{aligned} \label{eq:frombelow}
\end{align}
We procced by showing that the integral over $U_r$ of the $n$ lower terms in \eqref{eq:frombelow}, all coming with a minus, is finite, but the integral of the first term is not. We enlarge the domain of integration to all $p \in \R^3$ and perform a change of variables in $p \rightarrow p +k_j$ to obtain an upper bound for the integral over one of these terms:
\begin{align*}
\int_{U_r}&  \frac{\abs{ v(k_j)}^2 \Xi_\mu(p,K)^{2 \kappa} \abs{\psi(p+k_{j},\hat{K}_{j})}^2}{L_\mu(p,K)^2} \ud p \ud K 
\\
&
\leq \int_{\R \times B_r^{n}} \abs{\psi(p,\hat{K}_{j})}^2 \int_{\abs{k_j} < r} \frac{\Xi_\mu(p-k_j,K)^{2 \kappa} }{L_\mu(p-k_j,K)^2 \abs{k_j}^{2 \alpha}} \ud p \ud K \, .
\end{align*}
Here $B_r$ denotes the ball of radius $r$ in $\R^3$. Specifying to $\Xi_\mu = L_\mu$, we can bound the $k_j$-integral, using the fact that $\kappa<1$ and $\mu \geq 1$, by $\int_{\abs{k_j} < r} \abs{k_j}^{-2 \alpha} \ud k_j$. This is clearly finite since $\alpha < d/2$ by hypothesis. For $\Xi_\mu = \Omega_\mu$ and $\kappa \rightarrow \eta$ we bound $\Omega_\mu(K)^{2 \eta} L_\mu(p-k_j,K)^{-2} \leq \Omega_\mu(K)^{2 (\eta-1)} \leq 1$ and conclude in the same way.

To bound the integral over the first term in \eqref{eq:frombelow} from below, we use the assumption $\abs{v(k)} \geq c (1+\abs{k}^\alpha)^{-1}$ and the fact that $ \omega(k) \leq \abs{k}^{\beta}+ \mathsf m$ implies $\Omega(\hat{K}_1) \leq C$ for some constant on $U_r$:
\begin{align}
&\int_{U_r} \frac{\abs{ v(k_1)}^2 \Xi_\mu(p,K)^{2 \kappa}\abs{\psi(p+k_{1},\hat{K}_{1})}^2}{ L_\mu(p,K)^2}  \ud p \ud K 
\nonumber \\
&
\geq c \int_{B_r \times B_r^{n}} \abs{\psi(p,\hat{K}_{1})}^2 \int_{\R^3} \frac{  \Xi_\mu(p-k_1,K)^{2 \kappa}}{(1+\abs{k_1}^\alpha)^{2} ((p-k_1)^2+ k_1^\beta+ \mathsf{m}+ C)^2}  \ud p \ud K \, . \label{eq:integralbr}
\end{align}
When $\Xi_\mu = L_\mu$, we bound the integral over $k_1$ from below by
\begin{align*}
\int_{\R^3} \frac{(p-k_1)^{4 \kappa}}{(1+\abs{k_1}^\alpha)^{2} ((p-k_1)^2+ k_1^\beta+ \mathsf{m}+ C)^2}  \ud k_1 \, ,
\end{align*}
which does not converge for any fixed $p \in \R^3$ and $\kappa \geq (2-D)/4$. The same is true if $\Xi_\mu = \Omega_\mu \geq \abs k^\beta $ and $\eta \geq (2-D)/(2 \beta)$. Because $\psi \uppar n \neq 0$, we can choose $r>0$ such that the integral \eqref{eq:integralbr} is infinite. This proves the claim. \qed
\end{proof}

  \subsection{Self-Adjointness: Proof of Theorem~\ref{mainthm}}
 \label{subsectproofsmain}
 \setcounter{equation}{0}
At first we have to make sure that the construction sketched in the introduction is in fact possible in our case. We start by observing that the lower bound $c(1+\abs{k}^\alpha)^{-1} \leq \abs{v(k)}$ and the restriction $\alpha <3/2$ implies that $v \notin \Lz$. Thus, by ~\cite[Lem.~2.2]{nelsontype}, $\ker a(V)$ is dense in $\hilb$ and the adjoint of $L_\mu \big\vert_{\ker a(V)}=L_{\mu,0}$ is well defined. Using the fact that $G_\mu$ maps into $\ker L_{\mu,0}^*$, we arrive at the representation \eqref{eq:symmetricformofH}, which we repeat for the convenience of the reader:
\begin{align*}
H =  (1-G_\mu)^* L_\mu (1-G_\mu) + T^\mu - \mu \, .
\end{align*}

As has been discussed already in the introduction, it is necessary to prove infinitesimal boundedness of $T^\mu$ with respect to the self-adjoint operator $H_0^\mu$ (see Corollary~\ref{cor:Hnull}) for some $\mu \geq 1$. Then we can conclude with Kato-Rellich.  We will not aim at proving new results about $D(T)$ but instead recall that $u(s) := \frac{\beta }{2} s - \frac{D}{2}$ and cite the existing ones.
 
\begin{lem}[\textbf{Lemma~3.6 of~\cite{nelsontype} }]
\label{lem:DposTdiag}
Assume $D\geq 0$. Then for any $\eps>0$ the expression $T^\mu_\ud$ given by \eqref{eq:Tddef} defines a symmetric operator on the domain $D(L^{\max(\eps,D/2)})$ for any $\mu \geq 1$. 
\end{lem}

 \begin{lem}[\textbf{Lemma~3.8 of~\cite{pseudorel18} }]
\label{lem:T2mainlem}
Assume $D \geq 0$. Then, for all $s > 0$ such that $u(s) < 1$ and $0 < u(u(s))$, the operator $T^\mu_{\uod}$, defined in \eqref{eq:Toddef}, is bounded from $ D(N^{\max(0,1-s)} \Omega^{s-u(s)})$ to $\hilb$ and is symmetric on this domain for any $\mu \geq 1$.
\end{lem}
In order to apply the result of Lemma~\ref{lem:T2mainlem}, we clearly have to restrict to $s \geq 1$ as usual. 

\begin{proof}[\textbf{Proof of Theorem~\ref{mainthm}}]
Decompose into diagonal and off-diagonal terms $T^\mu  = T^\mu_{\ud} + T^\mu_{\uod}$. Due to Lemma~\ref{lem:DposTdiag}, we have a bound $\norm{T^\mu_{\ud} \psi} \leq \norm{L_\mu^{\max(\eps,D/2)} \psi}$. As long as $\mu$ is greater than some $\mu_0$ and $D < 2/3$, the second part of Corollary~\ref{cor:infibounds2} implies that the diagonal part of the operator is infinitesimally bounded by $H^\mu_0$. To proceed analogously for the off-diagonal part we need that for $s \geq 1$ Lemma~\ref{lem:T2mainlem} is applicable, so necessarily
\begin{align}
  \label{eq:condi4}
& u(s)<1
\\
  \label{eq:condi3}
& u(u(s))>0
 \, .
 \end{align}
In this way we can bound the norm of $T_\uod^\mu \psi$ by the norm of $\Omega_\mu^{s-u(s)} \psi$. Then we would like to conclude the infinitesimal boundedness by setting $\eta=s-u(s)$ in Corollary~\ref{cor:infibounds2}. To do so, we have to make sure that
 \begin{align} 
 \label{eq:condi11}
& s-u(s) < 1 , \\
 \label{eq:condi12}
& s-u(s) < \frac{2-D}{2 \beta} 
\end{align}
These four condition can be converted into bounds on $D$ that depend on $\beta$ and $s$:
  \begin{align*} 
  D> \beta s -2  &=: f_{\ref*{eq:condi4}}(s)
  \\
  D < s \frac{\beta^2}{\beta + 2} &=: f_{\ref*{eq:condi3}}(s)
 \\
 D < 2- s( 2-\beta)&=:f_{\ref*{eq:condi11}}(s)  \\
 D < \frac{2 - s \beta(2  - \beta ) }{\beta+1}&=:f_{\ref*{eq:condi12}}(s) \\
 \end{align*}
 If $\beta=2$, we choose $s=1$ and $D \in (0,2/3)$ to satisfy all four conditions. For $\beta \in (0,2)$, we assume $D \geq 0$ and set
 \begin{align}
 \label{eq:defofF}
 F:=\min_{i =\ref*{eq:condi11},\ref*{eq:condi12},\ref*{eq:condi3}} f_i: [1,2/\beta) \rightarrow \R \, .
 \end{align}
On this interval $[1,2/\beta)$, the Condition~\eqref{eq:condi4} is always satisfied and the Lemma~\ref{lem:minF} below completes the proof of Theorem~\ref{mainthm} because it confirms the upper bound on $D$. \qed
 \end{proof}

\begin{lem}
\label{lem:minF}
Let $F$ be as defined in~\eqref{eq:defofF}. For $\beta \in (0,2)$ it holds that 
\begin{align*}
\max_{s \in [1,2/\beta)} F(s) =  \begin{cases} \frac{\beta^2}{2} & \beta \in (0, 2(\sqrt{2}-1))  \\
\frac{2 \beta}{\beta+4}  & \beta \in [2(\sqrt{2}-1),\sqrt{5}-1) \\ \frac{\beta^2 -2 \beta +2   }{\beta+1} & \beta \in [\sqrt{5}-1,2) \, . \end{cases}
\end{align*}
\end{lem}
\begin{proof}
Closing the interval at the right endpoint we conclude that the supremum is attained, and we denote the point where this happens by $s_*$.
All functions $f_i$ are affine functions on $[1, 2/\beta]$. The fact that $(2-\beta) \geq 0$ implies that $f_{\ref*{eq:condi11}}$ and $f_{\ref*{eq:condi12}}$ are non-increasing whereas $f_{\ref*{eq:condi3}}$ is clearly increasing. Thus we have
\begin{align}
\label{eq:formulaforF}
F(s)  = \begin{cases} 
f_{\ref*{eq:condi3}}(s) & f_{\ref*{eq:condi3}}(s) < \min(f_{\ref*{eq:condi11}}(s),f_{\ref*{eq:condi12}}(s)) \\ \min(f_{\ref*{eq:condi11}}(s),f_{\ref*{eq:condi12}}(s))  & f_{\ref*{eq:condi3}}(s) \geq \min(f_{\ref*{eq:condi11}}(s),f_{\ref*{eq:condi12}}(s))
\, . \end{cases} 
\end{align}
If $\beta \geq \sqrt{5}-1$ then  $f_{\ref*{eq:condi12}}(1) \leq \min(f_{\ref*{eq:condi11}}(1),f_{\ref*{eq:condi3}}(1))$. This however implies that it holds that $F(s)=\min(f_{\ref*{eq:condi11}}(s),f_{\ref*{eq:condi12}}(s))$ and consequently $s_* = 1$. We can also conclude that 
%
\begin{align*}
F(s_*) = \min(f_{\ref*{eq:condi11}}(1),f_{\ref*{eq:condi12}}(1))=f_{\ref*{eq:condi12}}(1) = \frac{\beta^2 -2 \beta +2   }{\beta+1}\, .
\end{align*}
Now consider the case where $\beta<\sqrt{5}-1$. Because $f_{\ref*{eq:condi3}}(1) < f_i(1)$ for $\beta \in (0,\sqrt{5}-1)$, we observe that the first case of \eqref{eq:formulaforF} is never empty. Consequently 
\begin{align*}
s_*&:=\left\lbrace s\in[1,2/\beta] \big \vert f_{\ref*{eq:condi3}}(s) = \min(f_{\ref*{eq:condi11}}(s),f_{\ref*{eq:condi12}}(s)) \right\rbrace  = \min_{i=\ref*{eq:condi11},\ref*{eq:condi12}} \left\lbrace s\in[1,2/\beta] \big \vert f_{\ref*{eq:condi3}}(s) = f_i(s) \right\rbrace  
\end{align*}
and of course $F(s_*) = f_{\ref*{eq:condi3}}(s_*)$. We find that 
\begin{align*}
f_{\ref*{eq:condi3}}(s) &= f_{\ref*{eq:condi11}}(s) \ \iff \ s = s_{\ref*{eq:condi11}} := \frac{\beta+2}{2} 
\\
f_{\ref*{eq:condi3}}(s) &= f_{\ref*{eq:condi12}}(s) \ \iff \ s = s_{\ref*{eq:condi12}} := \frac{2(\beta+2)}{\beta(\beta+4)} 
\end{align*}
and, because both $s_{\ref*{eq:condi11}}$ and $s_{\ref*{eq:condi12}}$ lie in the interval $[0,2/\beta]$, that means
\begin{align*}
s_*& = \min(s_{\ref*{eq:condi11}}, s_{\ref*{eq:condi12}}) = \begin{cases} s_{\ref*{eq:condi11}} & \beta \in (0, 2(\sqrt{2}-1)) \\
s_{\ref*{eq:condi12}} & \beta \in [2(\sqrt{2}-1),\sqrt{5}-1) \end{cases} 
\\
&
= \frac{\beta+2}{2} \begin{cases} 1  & \beta \in (0, 2(\sqrt{2}-1)) \\
\frac{4}{\beta(\beta+4)}  & \beta \in [2(\sqrt{2}-1),\sqrt{5}-1) \, . \end{cases} 
\end{align*}
Insert this into $f_{\ref*{eq:condi3}}$ and note that $s_*<2/\beta$. This yields the desired expression for $\max_{s \in [1,2/\beta)} F(s)$. \qed
\end{proof}

\section{Concluding Remarks}
We would like to adress two points that have not been discussed so far. We have not said anything yet about the connection of the IBC approach to renormalisation procedures in the massless case. In~\cite[Prop.~3.4]{pseudorel18}, it is shown that for quite general massive models the cutoff operator plus renormalisation constant $H_\Lambda+E_\Lambda$ converges in norm resolvent sense to the Hamiltonian $H$. In this cutoff operator, the form factor in the formal expression $L+a(V)+a^*(V)$ is replaced by $\chi_\Lambda v$ for the characteristic function $\chi_\Lambda$ of a ball of radius $\Lambda$. 

As we will argue in the following, such a result does also hold in the case of massive or massless models if Condition~\ref{maincond} is satisfied. Denote by $G_{\mu,\Lambda}$ and $T^\mu_\Lambda$ the corresponding operators with $ v $ replaced by $\chi_\Lambda v$. The parameter $\mu\geq 1$ is chosen as large as necessary and fixed. For the proof of norm resolvent convergence, convergence of $G_{\mu,\Lambda}$ in the $\hilb$-norm (to $G_\mu$) is needed. As long as $u(1) \in (0,1)$, this follows in the massless case exactly as in \cite[Prop.~4.4]{pseudorel18} by explicitly computing symmetric decreasing rearrangements. To prove convergence of the STM-operator $T^\mu$, it is convenient to decompose again into diagonal and off-diagonal parts. Using~\cite[Lem.~3.6 and 3.8]{pseudorel18}, we can prove convergence of $T^\mu_{\ud,\Lambda}+E_\Lambda$ on $D(L^\kappa)$ and of $T^\mu_{\uod,\Lambda}$ on $D(\Omega^\eta)$ for some $\kappa,\eta$. It turns out that $\kappa,\eta$ are such that $T^\mu_\Lambda+E_\Lambda \rightarrow T^\mu$ on $D(H)$. This would complete the proof of norm resolvent convergence.

Although the case of a single particle was considered in this contribution in order to keep the notation simple, the case of $M>1$ particles could be included in the analysis as well. This is because when bounding norms of $G_\mu \psi$ from above, the relevant estimates are the same for $M=1$ and $M>1$. For bounds from below, as in Section~\ref{proofofpropreverse}, one has to take care of some more cross-terms because the domain of integration is chosen to be not symmetric under exchange of particles. It is only the $T$-operator where a significant difference occurs. The off-diagonal part of $T$ consists for $M>1$ of additional terms, which are called $\theta$-terms in \cite{nelsontype}. They are however bounded on $D(L^{\max(\eps,D/2)})$ for any $\eps>0$, exactly as the diagonal part of $T$, see \cite[Lem.~3.7]{pseudorel18}. In the context of the above analysis, these $\theta$-terms can therefore be put together with $T^\mu_\ud$ and pose almost no constraints on the allowed pairs $(\beta,D)$.

\newcommand{\etalchar}[1]{$^{#1}$}
\providecommand{\bysame}{\leavevmode\hbox to3em{\hrulefill}\thinspace}
\providecommand{\MR}{\relax\ifhmode\unskip\space\fi MR }
\providecommand{\MRhref}[2]{%
  \href{http://www.ams.org/mathscinet-getitem?mr=#1}{#2}
}
\providecommand{\href}[2]{#2}

\end{document}